\documentclass[11pt]{article}
\usepackage{pifont,pst-all,wrapfig}
\definecolor{lightgray}{rgb}{0.5,0.5,0.5}
\newcommand{\numtiny}[1]{\footnotesize\raisebox{0.1ex}{#1}}
\newtheorem{theorem}{Theorem}[section]
\newtheorem{lemma}[theorem]{Lemma}
\newtheorem{corollary}[theorem]{Corollary}

\newtheorem{conjecture}{Conjecture}[section]
\newtheorem{observation}{Observation}[section]
\newcommand{\qed}{\hfill\ding{113}}
\newcommand{\Dij}{\textsf{Dij}}
\newcommand{\ARead}{\textsf{ARead}}
\newcommand{\AReadk}{\textsf{ARead}(\textit{k})}
\newcommand{\AWrite}{\textsf{AWrite}}
\newenvironment{proof}{\begin{trivlist}
\item[\hspace{\labelsep}{\bf\noindent Proof: }]}{\qed\end{trivlist}}

\def\itp#1{(\textit{#1}\/)}
\newcommand{\doublespace}{\addtolength{\baselineskip}{.25\baselineskip}}
\begin{document}
\doublespace

\title{Safe Register Token Transfer in a Ring}
\author{Ted Herman, University of Iowa}
\maketitle

\begin{abstract}
A token ring is an arrangement of $n$ processors
that take turns engaging in an activity which 
must be controlled.  A token confers the right 
to engage in the controlled activity.  
Processors communicate with neighbors
in the ring to obtain and release a token. 
The communication mechanism investigated in
this paper is the safe register abstraction, 
which may arbitrarily corrupt a value that 
a processor reads when the operation reading
a register is concurrent with an write operation
on that register by a neighboring processor.
The main results are simple protocols for quasi-atomic
communication, constructed from safe registers.  
A quasi-atomic register behaves atomically 
except that a special $\perp$ value may be 
returned in the case of concurrent read and write 
operations.  Under certain conditions that 
constrain the number of writes and registers, 
quasi-atomic protocols are adequate substitutes
for atomic protocols.  The paper demonstrates how
quasi-atomic protocols can be used to implement 
a self-stabilizing token ring, either by using 
two safe registers between neighboring processors
or by using $O(\lg n)$ safe registers between 
neighbors, which lowers read complexity.     
\end{abstract}


\textsc{CR Categories:} \emph{H.3.4 Distributed
Systems; D.1.3 Concurrent Programming; 
\\ \hspace*{6em} D.4.1 
Process Management (Synchronization)} 

\textsc{TR Number:}  TR-11-01 Department of Computer Science,
University of Iowa

\begin{flushright}
\emph{The fundamental task in computing is to implement \\
higher-level operations with lower-level ones.} \\ 
--- Leslie Lamport \cite{Lam86c}
\end{flushright}

\section{Introduction}
\label{section:intro}

Among the many qualitative dimensions characterizing distributed
computing are time and communication types.  Time, whether in 
processing rate, duration, or delay of communication operations,
may be modeled synchronously or asynchronously;
types of communication include transient (message passing),
persistent (shared memory) \cite{Lam86b}, and rendezvous \cite{Hoa78}.       
Asynchronous models pose the most challenging problems for reasoning
about program properties, particularly when failures are considered. 
The model of shared objects with prescribed operations 
captures the essence of persistent communication between 
asynchronous processes.  The most primitive 
type of shared object is a register with only two operations, read
and write.  

Emerging large scale platforms, notably multicore architectures 
and cloud computing facilities, motivate relaxed consistency 
operations, nonblocking semantics, and speculative or probabilistic
approaches.  Register abstractions are of potential interest for 
two reasons:  first, registers are wait-free operations which 
are meaningful at both low-level (machine architecture) and 
high-level design (manipulating key-value pairs);  second, the 
literature on registers has explored numerous models of concurrency
restriction and degraded semantics, finding constructions that
overcome deficiencies of unreliable read operations.

Register properties can be axiomatized 
\cite{Mis86,Lam86a,Lam86b}, with several choices 
for behavior during concurrent operations;  different choices lead
to stronger or weaker register types.  The strongest 
type is an atomic register, and the weakest type is 
a safe register.  Atomic registers are most useful for applications, 
because they simplify reasoning in the face of concurrent execution.
Safe registers are most convenient for implementors, because they
have minimal requirements on behavior under concurrent execution.   
A significant literature of protocols and constructions explores
how atomic behavior can be derived from safe registers or other 
shared objects with weak semantics.  Such constructions are typically
complex, from a resource standpoint (many low-level registers
needed to implement a higher-level atomic one) 
or from a verification standpoint.  

\paragraph{Contributions.}
Protocols presented in this paper show how a self-stabilizing
token ring can be implemented using safe registers, which are 
the weakest type in Lamport's register hierarchy \cite{Lam86a,Lam86b}. 
Previous work showed that regular and safe registers suffice 
for communication in a self-stabilizing token ring \cite{DH01};  
the contributions of the new protocols are an improved validation
framework and a construction that uses two safe registers rather
than $O(\lg n)$ safe registers between neighbors.
The safe register protocols for read and write
operations are simple, thanks to the closed-loop nature
of the stabilizing token ring, which inherently limits concurrency.
One may question whether exploiting a concurrency-limiting property
is interesting, since the point of wait-free operations is to allow
unrestricted concurrency.  In fact, many high-level tasks have 
some sequential or concurrency-limiting properties, and it is sensible
to exploit such properties if they simplify lower-level design.  
Moreover, retaining wait-free behavior of low-level operations can
benefit implementation designs (which might use speculation, caching,
and other ideas) even when higher-level tasks are sequential.      

\paragraph{Organization.}  Section \ref{section:prelim} briefly 
reviews terminology for register abstractions and the token
ring, and Section \ref{section:dijatom} casts self-stabilizing
token passing in terms of atomic registers.  Then Section 
\ref{section:dijquasi} introduces quasi-atomicity and a protocol
implementing quasi-atomic operations using safe registers 
(subsection \ref{section:qa}).  Two following sections, 
\ref{section:tworeg} and \ref{section:logreg}, present  
self-stabilizing adaptations of the token ring using the 
quasi-atomic constructions.  Discussion wraps up the 
paper in Section \ref{section:discuss}.

\section{Preliminaries}
\label{section:prelim}

This section informally reviews terminology of 
registers, constructions from registers, and a self-stabilizing
token ring protocol.  An atomic register is a shared object with
two methods, read and write.   The time between invocation of 
a register method and its response can be arbitrary in duration, 
which allows for interleaving of steps from different processors
in an execution.  Two operations are considered to be 
concurrent if the invocation of one occurs in the interval between
invocation and response of the other.   
Formal verification of protocols consists of
mapping an interleaved execution to a linearized history of 
processor steps and register operations in such that each invocation
of a register method is immediately followed by its response in
the history;  the verification arguments in this paper are 
informal, reasoning at the level of operation properties rather
than constructing linearized mappings.  Additional nomenclature
is given in Section \ref{section:dijquasi} 
for reasoning about operation intervals.  In executions without
concurrent register operations, behavior of read is unambiguous:
the response to any read is the value most recently written to
that register.    

To describe atomic behavior
operationally, consider a write invocation $\textsf{W}(x)$ on 
a register \textsf{R} which contains the value $y$ prior to 
$\textsf{W}(x)$, where $x\neq y$.  Value $y$ is called the 
\emph{old} value, and $x$ is the \emph{new} value.  A register
is atomic if any read not concurrent with a write responds with
the most recently written value, and read operations concurrent
with a write responds with either $y$ or $x$, subject to the 
constraint that once a read returns $x$, any subsequent
read also returns $x$.    

A regular register weakens atomicity somewhat:  a read concurrent
with a write may return the old or new value arbitrarily.  A 
safe register weakens atomicity further, only guaranteeing that
a read concurrent with a write returns some value in the domain
of the register (binary, $m$-bit integer, or whatever the capacity
is given for the register).  It is perhaps surprising that safe 
registers could be useful, until one sees that the definitions 
of regular and safe collapse for the case of a single-bit register,
provided that is only written when the current value needs to be changed.   
One way to specify an atomic register is add a constraint to a regular
register:  an atomic register is a regular register without 
new-old inversion, that is, the old value is not returned once 
the new value has been returned in a sequence of read operations
concurrent with a write. 

Register properties become more complex when many processors read
and write the same register.  The notation $m\textrm{W}n\textrm{R}$ 
indicates that $m$ processors, called \emph{writers}, and $n$ 
processors, called \emph{readers}, may concurrently have operations
on the same register.  For the token ring protocols in this paper, 
communication is confined to a ring in which each processor only 
shares registers with neighbors in the ring;  moreover, communication
is unidirectional, because a token is consistently passed from 
each processor to only one other.  The type of register could be 1W1R, 
except that self-stabilization invalidates the assumption that 
a writer's internal state correctly estimates the value of the 
register to be written---such an assumption is important to avoid
writing except when needed to change a value.  Therefore, the protocols
use 1W2R registers, so that the processor writing can
also read that register.   

\section{\Dij\ using Registers}
\label{section:dijatom}

The vehicle for demonstrating quasi-atomic registers is a 
self-stabilizing token ring protocol \cite{Dij73,Dij74}.  This 
protocol is a simple construction that depends on atomic 
communication for the self-stabilization property.  The protocol
was originally expressed as a ring of processes communicating
through shared state variables;  subsequent work adapted the 
protocol to a register model of communication 
\cite{Dol00}\footnote{Lamport also introduced an 
adaptation of \Dij\ to a common shared memory model in \cite{Lam86d}.}.    
A register-based adaptation of this famous protocol is 
shown in Figure \ref{fig:orDij}.  We call this the \Dij\ protocol
in the remainder of the paper.

\begin{figure}[ht]
\hrule\vspace*{2ex}\par
\begin{tabbing}
XXXXXXXX \= xxx \= xxx \= xxx \= xxx \= xxx \= xxx \= xxx \= \kill
\>\numtiny{1}\> $\Dij_i(K)$: \\
\>\numtiny{2}\>\> local variables $x$, $y$ \\
\>\numtiny{3}\>\> do \textit{forever} \\
\>\numtiny{4}\>\>\> $y \;\leftarrow\;$ read output \textsf{R} of $p_{i\ominus 1}$ \\
\>\numtiny{5}\>\>\> if $i\neq 0 \;\wedge\; x\neq y$ then \\
\>\numtiny{6}\>\>\>\> $x \;\leftarrow\; y$  \\
\>\numtiny{7}\>\>\>\> \emph{critical section}  \\
\>\numtiny{8}\>\>\> if $i=0 \;\wedge\; y=x$ then \\
\>\numtiny{9}\>\>\>\> $x \;\leftarrow\; (x+1)\bmod K$ \\
\>\numtiny{10}\>\>\>\> \emph{critical section}  \\
\>\numtiny{11}\>\>\> write $x$ to output register \textsf{R} 
\end{tabbing} 
\hrule
\caption{register-based \Dij\ protocol for $p_i$}
\label{fig:orDij}
\end{figure}

Figure \ref{fig:orDij} describes the behavior of processor $p_0$
(lines 8-10) and the behavior of processors $p_1$--$p_{n-1}$ (lines 5-7). 
The ring uses unidirectional communication, as each processor 
reads a register (line 4) written by the previous processor in
the ring (writing occurs on line 11).  The token abstraction 
is embodied by conditions on lines 5 and 8, which allow a processor
to execute a ``critical section'' representing some activity to be
controlled, like mutual exclusion.  Terms $p_{i\ominus 1}$
and $p_{i\oplus 1}$ denote the previous and next processors, with 
respect to $p_i$, in the ring.  The registers are supposed to be atomic.
Variables $x$, $y$, and the registers may have arbitrary initial 
values, however the domain of all variables and registers is 
confined to the set $\{\,i\;|\; 0\leq i<K\}$, where $K$ is some 
given constant satisfying $K>2n$.  The proof of self-stabilization
for \Dij\ is typically shown by defining a subset of the state-space      
of the ring of processors called the \emph{legitimate} set, showing
that this set is closed under execution (each successor of a 
legitimate state is legitimate), and that it satisfies safety and 
liveness properties (the token perpetually advances in the ring, 
and there is always a single token).  Convergence from an arbitrary
initial state consists of showing (by contradiction) the absence 
of deadlock, \emph{e.g.} that $p_0$ must infinitely often  
execute line 9, and that eventually the assignment on line 9 
obtains a value different from that in any other variable or 
register throughout the ring.  

We suppose in the sequel that the reader is familiar with 
stabilization arguments \cite{Var00,Dol00} for \Dij, 
and confine our task to replacing the atomic registers 
used in Figure \ref{fig:orDij} by 
constructions using safe registers.  This paper does not attempt
to settle fundamental questions about possibility or impossibility
of token circulation using safe registers;  for instance, we do 
not explore the space of algorithms that communicate bidirectionally
between neighbors in the ring, nor do we investigate probabilistic 
register constructions.  Rather, we take the \Dij\ protocol as the 
given structure to implement, and consider how it can be adapted
to safe register communication.

Transforming \Dij\ from its original shared state model to using
atomic link registers is straightforward, and using regular 
registers instead of atomic ones isn't a challenging problem. 
Safe registers, however, require more interesting protocols because
these registers have weak concurrency properties.  
The only guarantee by a safe register is that a read not concurrent with
a write will return the most recently written value.  
A read concurrent with a write can return any value in the 
register's domain, even if the value
being written is already equal to what the register contains.
Two of the difficulties in constructing a transformation are neatly
summarized in the following conjectures. 
\begin{conjecture} \label{conj:1reg} 
Algorithm \Dij\ cannot be implemented using only one 
safe register between $p_i$ and $p_{i\oplus 1}$.
\end{conjecture}
The intuition for this conjecture is that a processor with only 
one safe register must write to that register in some case, and
the reader can have unboundedly many reads concurrent with such 
a write operation, each resulting in an arbitrary value.   
\begin{conjecture} \label{conj:1w1r} 
Algorithm \Dij\ cannot be implemented using only 
1W1R safe registers between $p_i$ and $p_{i\oplus 1}$.
\end{conjecture}
The intuition for this second conjecture is that a processor
cannot ascertain the value of its output register, and therefore
must continually rewrite it, but doing so admits the possibility
of having every register read being concurrent with a write, 
which would defy progress.
\par
If conjecture \ref{conj:1reg} holds, then any transformation will
have $p_i$ write more than one register that $p_{i\oplus 1}$ reads.
Section \ref{section:tworeg} provides a transformation using two
registers for each processor in the ring, which would be optimal
if the conjecture holds.  If conjecture \ref{conj:1w1r} holds, 
then any transformation will allow that the writers of registers 
can also read the values of the
registers they write:  these safe registers are 1W2R registers.  
Since no processor can read and write the 
same register concurrently, any read by $p_i$ of its own output 
register is trivially atomic.  

\section{Quasi-Atomic Registers}
\label{section:dijquasi}

To the standard terminology mentioned in previous sections, 
a variation of the atomicity property is used in protocols of later sections.  
\emph{Quasi-atomic} behavior differs
from atomic behavior only in that a read operation may return the
exception value $\perp$, indicating a ``busy'' condition where the
reader should retry the operation.  Similarly, let a \emph{quasi-regular}
register differ only from a regular register by allowing a read
to respond with $\perp$.  In the absence of concurrency, the special $\perp$
value is not returned by a read.

Reasoning about nonatomic register operations is often explained 
with diagrams and ordering relations.  Diagrams illustrate how 
register operations have duration, and how the time intervals of
the operations are related.  Figure \ref{fig:overlaps} 
shows a typical case of two consecutive write operations, $W$ and $W'$,
both due to some processor $p$ writing to the same register, and 
two consecutive read operations, $R$ and $R'$,
of that register by another processor $p'$.  The figure shows that
$W$ ends before $R'$ begins;  thus $W$ \emph{precedes} 
$R'$, written $W\prec R'$.  We write $W\preceq R$ if   
$W$ starts before $R$ starts:  either $W\prec R$ or 
the two operations are concurrent.  In the case of  
Figure \ref{fig:overlaps}, $W\preceq R$ and $R\preceq W'$.
Protocols for high-level operations usually include numerous register 
operations by each processor.  For instance, the two write operations
of processor $p$ in Figure \ref{fig:overlaps} could be due to some
higher-level procedure call, which has a duration spanning the 
intervals of $W$ and $W'$.  The interval from the start of $W$ to 
the end of $W'$ is said to \emph{contain} the interval of $R$, 
because $W$ begins before $R$ and $W'$ ends after $R$ ends.   

\begin{figure}[ht]
\hrule\vspace*{2ex}\par
\begin{pspicture}(-0.5,-1.0)(8,1)
\begin{psmatrix}[rowsep=5pt,colsep=28pt]
$p$ \quad  & ~ & ~ & ~ & ~ & ~ & ~ & ~ & ~ & ~ & ~ & ~ & ~ \\
$p'$ \quad & ~ & ~ & ~ & ~ & ~ & ~ & ~ & ~ & ~ & ~ & ~ & ~ 
\ncline[arrows=|-|]{1,2}{1,4}^{$W$}
\ncline[arrows=|-|]{1,5}{1,9}^{$W'$}
\ncline[arrows=|-|]{2,3}{2,6}_{$R$}
\ncline[arrows=|-|]{2,7}{2,11}_{$R'$}
\end{psmatrix}
\end{pspicture}
\hrule
\caption{write and read operations may overlap}
\label{fig:overlaps}
\end{figure}
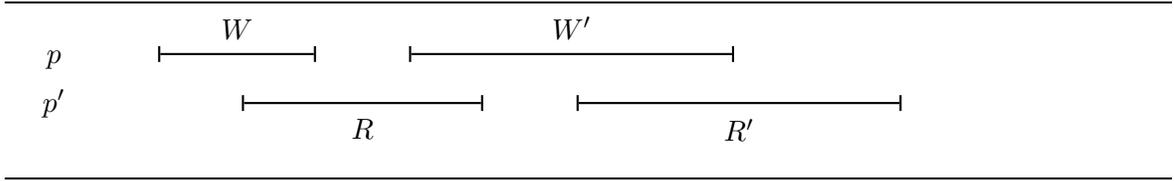

Some simple algebra on the relations between register operation intervals
aids in reasoning about register protocols.  The $\prec$ relation is 
transitive and containment is also transitive.  Concurrency 
is not transitive:  $A$ being concurrent with $B$ and $B$ being 
concurrent with $C$ does not imply that $A$ and $C$ are concurrent.  
The $\preceq$ relation is not transitive, 
however some combinations are transitive, for 
instance $A\preceq B \;\wedge\; B\prec C \;\Rightarrow\; A\preceq C$
holds because $A$ begins before $B$ ends, hence $A$ begins 
before $C$ begins, which implies $A\preceq C$.  The following 
inference about containment is used later in this section.   
\begin{observation}
Suppose $A\preceq B_0$, $B_j\prec B_{j+1}$ for $0\leq j<m$, 
and $B_m\preceq A'$.  Then the time interval from the start of $A$
to the end of $A'$ contains an interval that begins with 
some time instant in $B_0$ and ends with some point in $B_m$.  
\end{observation}

\subsection{Duplicate Writes and $k$-Scan Reads}
\label{section:qa}

We propose here a protocol that, under certain conditions, 
transforms safe registers to quasi-atomic behavior for communication
from $p_i$ to $p_{i\oplus 1}$.  The proposed protocol consists of an 
\AWrite\ procedure invoked by $p_i$ and an \AReadk\ 
procedure invoked by $p_{i\oplus 1}$.  Figure \ref{fig:dwkr} 
shows the two procedures, which use a pair of 1W2R registers between
$p_i$ and $p_{i\oplus 1}$.  An \AWrite(\textit{val}) invocation 
writes \textit{val} to registers $\textsf{R}_a$ and $\textsf{R}_b$, 
but only if these registers do not already both contain \textit{val}.
An \AReadk\ invocation reads both of these registers $k$ times
in succession, returning $\perp$ if not all of the read operations 
yield the same value, and otherwise returning the (unanimous) 
value from the registers.  On one hand, 
value $\perp$ indicates a reading failure,
that is, $\perp$ is returned when it is known that \AReadk\ could
not return a value with atomic read semantics.  On the other hand,
when \AReadk\ does not return $\perp$ we cannot be sure that
the returned value is an atomic read of the latest value from an
\AWrite\ operation.  The following lemma finds a condition 
for which \AReadk\ is quasi-atomic.

\begin{figure}[ht]
\hrule\vspace*{2ex}\par
\begin{tabbing}
XXXXXXXX \= xxx \= xxx \= xxx \= xxx \= xxx \= xxx \= xxx \= \kill
\>\numtiny{1}\> \AWrite(\textit{val}): \\
\>\numtiny{2}\>\> local variables $A$, $B$ \\
\>\numtiny{3}\>\> read from $\textsf{R}_a$ into $A$ \\ 
\>\numtiny{4}\>\> read from $\textsf{R}_b$ into $B$ \\ 
\>\numtiny{5}\>\> if $A=B=\textit{val}$ then return \\
\>\numtiny{6}\>\> else \\
\>\numtiny{7}\>\>\> write \textit{val} to $\textsf{R}_a$ \\ 
\>\numtiny{8}\>\>\> write \textit{val} to $\textsf{R}_b$ \\ 
\>\numtiny{9}\>\>\> return \\ \\
\>\numtiny{10}\> \AReadk: \\
\>\numtiny{11}\>\> local array $A[k]$, $B[k]$ \\
\>\numtiny{12}\>\> for $i=1$ to $k$:  \\
\>\numtiny{13}\>\>\> read from $\textsf{R}_a$ into $A[i]$ \\
\>\numtiny{14}\>\>\> read from $\textsf{R}_b$ into $B[i]$ \\
\>\numtiny{15}\>\> if all of $A[..]$ and $B[..]$ have same value \\
\>\numtiny{16}\>\>\> then return $A[1]$ \\
\>\numtiny{17}\>\> else return $\perp$ 
\end{tabbing} 
\hrule
\caption{duplicate write, $k$-scan read protocol}
\label{fig:dwkr}
\end{figure}

\begin{lemma} \label{lemma:dwkr}
In any execution where $p_i$ invokes \AWrite\ at most $(k-1)$ times,
then every \AReadk\ invocation is a quasi-atomic read by $p_{i\oplus 1}$. 
\end{lemma}
\begin{proof}
The proof begins by showing that any \AReadk\ is quasi-regular, that is, 
it either returns the value of the registers prior to any \AWrite\ 
commencing, or the value of the registers after some \AWrite\ is 
finished and before the next \AWrite\ starts, or the value $\perp$. 
Then this argument is generalized to show that in any    
sequence of \AReadk\ invocations, no new-old inversion occurs.  
We show that any \AReadk\ is regular by contradiction, after 
first introducing a graph to represent the interaction between 
\AWrite\ and \AReadk\ operations on the safe registers.

To disambiguate \AWrite\ invocations that may have 
the same \textit{val} argument (see Figure \ref{fig:dwkr}),  
we assume that each \AWrite\ is invoked
to write a value distinct from all other (at most $k-2$) 
\AWrite\ invocations.  Giving each \AWrite\ 
a different input value from the previous \AWrite\ presents  
a worst case execution with regard to the number of low-level writes.
At the end of the proof, we examine
cases where this assumption does not hold.  

Consider a single write operation to a safe register and a possibly 
concurrent read operation on that register.  Three possibilities are
\itp{i} the read returns the \emph{old} value of the register (that is,
the value that the register holds prior to the write), \itp{ii} 
the \emph{new} value of the register (that is, the register's value 
after the write is complete), or \itp{iii} an arbitrary value returned
because the read operation is concurrent with the write operation.  
A sequence of $(k-1)$ \AWrite\ invocations produces a sequence 
of writes to $\textsf{R}_a$ and $\textsf{R}_b$ registers, 
which we denote as
\begin{eqnarray} \label{eqn:awrite0}
W^1_a \; 
W^1_b \; 
W^2_a \; 
W^2_b \; 
\cdots \;
W^{k-1}_a \; 
W^{k-1}_b 
\end{eqnarray}
For cases \itp{i}--\itp{iii} the values of the register are of 
concern.  Instead of looking at the sequence of write operations,  
we therefore examine the sequence 
\begin{eqnarray} \label{eqn:awrite1}
u^1 \;
w^1_a \; 
v^1 \;
w^1_b \; 
u^2 \;
w^2_a \; 
v^2 \;
w^2_b \; 
u^3 \;
\cdots \;
u^{k-1} \;
w^{k-1}_a \; 
v^{k-1} \;
w^{k-1}_b 
u^k \;
\end{eqnarray} 
which distinguishes all possible \emph{situations} that  
read operations on registers $\textsf{R}_a$ and $\textsf{R}_b$ 
may encounter during an execution.  
Term $u^1$ represents the situation where no writing has begun.
Term $w^1_a$ represents the interval of $W^1_a$, which can yield an 
ambiguous value;
$v^1$ signifies that $W^1_a$ is finished, but $W^1_b$ 
has not started.  Term $w^1_b$ represents the interval of $W^1_b$.  
Term $u^2$ is the situation where $W^1_a$ and $W^1_b$ have finished,
but $W^2_a$ has not yet started. 
Any \AReadk\ operation induces a sequence
of read operations on $\textsf{R}_a$ and $\textsf{R}_b$,         
\begin{eqnarray} \label{eqn:aread0}
R^1_a \; 
R^1_b \; 
R^2_a \; 
R^2_b \; 
\cdots \;
R^{k}_a \; 
R^{k}_b 
\end{eqnarray}
The sequence of read operations (\ref{eqn:aread0}) is related to 
sequence (\ref{eqn:awrite1}).  A convenient portrayal of this 
relation is the following graph.  First, let the  
terms of (\ref{eqn:aread0}) be one set of vertices, and 
the terms of  (\ref{eqn:awrite1}) are another set of vertices.
The relation is given by adding edges between these
two sets to form a bipartite graph induced by 
values returned from read operations.
For example, if $R^2_a$ is concurrent with a write operation in the  
execution and returns a value different from $\textsf{R}_a$'s initial
content and different from any \textit{val} previously 
written to $\textsf{R}_a$,
then there is an edge between $R^2_a$ and some $w^j_a$ vertex.   
If instead $R^2_a$ reads the value between $w^1_a$'s completion and $w^2_a$
starting, there is an edge between $R^2_a$ and one of   
$\{v^1,w^1_b,u^2\}$.  We say that an $R$-vertex \emph{maps to} a 
$v$, $w$, or $u$ vertex according to the constructed graph. 
In addition to edges between vertices of (\ref{eqn:aread0}) and 
(\ref{eqn:awrite1}), let edges also be added to the graph between
successive items in each respective sequence:  $(u^1,w^1_a)$, 
$(w^1_a,v^1)$, \ldots, are edges;  and $(R^1_a,R^1_b)$, $(R^1_b,R^2_a)$, 
\ldots, are edges.  The resulting graph is planar:  the edges
mapping $R$-vertices to vertices from (\ref{eqn:awrite1}) do not
cross (cases (a)-(d) below explain this point).
 
With aid of the bipartite graph between reads and writer situations, 
we return the proof of the lemma, which is an implication, proved
here by contradiction.  A refutation of the lemma supposes an \AReadk\ returns 
an arbitrary non-$\perp$ value, that is, a value that 
does not correspond to any of $\{\,u^i\;|\;1\leq i\leq k\}$; 
terms  $\{\,v^i\;|\;1\leq i<k\}$ represent intermediate points where 
$\textsf{R}_a\neq\textsf{R}_b$, and \AReadk\ would return
$\perp$, giving a contradiction.  It follows that every safe-register 
read operation returns the same arbitrary value $x$, different from 
the value corresponding to any of $\{\,v^i\;|\;1\leq i<k\}$. 
Therefore each term of the form $R^j_a$ maps to a 
vertex in $\{\,w^i_a\;|\;1\leq i<k\}$,   
and each term of the form $R^j_b$ maps to a vertex in 
$\{\,w^i_b\;|\;1\leq i<k\}$.  Sets $\{\,R^j_a\;|\,1\leq j\leq k\}$  
and $\{\,R^j_b\;|\,1\leq j\leq k\}$ each have $k$ vertices, 
whereas $|\{\,w^i_a\;|\;1\leq i<k\}| = k-1$ and    
$|\{\,w^i_b\;|\;1\leq i<k\}| = k-1$.  Some elementary observations
about the ordering of read and write operations constrain mapping
from read operations to (\ref{eqn:awrite1}) vertices, as follows.   
\begin{itemize}
\item[(a)] $R^j_a$ and $R^j_b$ map to distinct vertices because the former
maps to a write of $\textsf{R}_a$ and the latter to a write of $\textsf{R}_b$.
\item[(b)] For $R^j_a$ and $R^j_b$, an edge from $R^j_a$ to $w^m_a$ implies that
the edge from $R^j_b$ to $w^n_b$ satisfies $n\geq m$ by the ordering
of the sequence of write operations.  
\item[(c)] For $R^j_a$ and $R^{\ell}_a$,
$\ell>j$, with $R^j_a$ mapping to $w^m_a$ and $R^{\ell}_a$ mapping to 
$w^n_a$, the sequential ordering of the read operations implies
$n\geq m$ (a similar observation holds for $R_b$ operations).    
\item[(d)] Observation (c) can be strengthened to $n>m$, because between 
any two read operations on $\textsf{R}_a$ there is a read operation
on $\textsf{R}_b$, and observations (a) and (b) constrain the mapping 
targets to be distinct.    
\end{itemize}
By induction, for any $h>j$, $R^h_a$ maps to a vertex distinct from 
the vertices that $R^j_a$ and $R^j_b$ map to.  Since the
number of $R$ vertices is $2k$ and the number of $w$ vertices is 
$2(k-1)$, the distinctness constraint mapping $R$ vertices to 
$w$ vertices implies a contradiction.  This contradiction shows
that any \AReadk\ returns a value that is either the initial 
value of the \textsf{R}-registers or a value that was written 
by some \AWrite\ operation preceding the \AReadk\ or concurrent
with the \AReadk\ operation.  If the value is due to an \AWrite\ 
preceding the \AReadk, then it must be the last such \AWrite, 
because safe registers return the most recently written value
in the absence of concurrency.  Therefore, the protocol is    
quasi-regular.  

Proof of quasi-atomicity consists of showing that ordered 
\AReadk\ invocations do not exhibit new-old inversion. 
Suppose that the sequence of arguments to the $(k-1)$ 
\AWrite\ operations is 
$x^1$, $x^2$, \ldots, $x^{k-1}$ (let $x^0$ be the initial
value of $\textsf{R}_a$ and $\textsf{R}_b$).  
Consider two \AReadk\ invocations $A$, $A'$, such that
$A'$ occurs after $A$, both with non-$\perp$ responses, and
$A'$ returns $x^i$ while $A$ returns $x^j$.  New-old inversion
occurs if $j<i$.  However, $j<i$ contradicts planarity of the 
graph.

The arguments above verify the proof obligation when each 
\AWrite\ has a distinct value;  we now consider executions
where \AWrite\ invocations get repeated values.  The 
simplest scenario is when consecutive \AWrite\ invocations have
the same value:  in such cases, repeated \AWrite\ invocations
are not effective, because line 5 of Figure \ref{fig:dwkr}
is an early exit.  Thus we focus on executions where repeated
\AWrite\ values are not consecutive.  Here, there can be ambiguity
in mapping low-level $R$-vertices to register situations.  However,
the behavior of \AReadk\ in Figure \ref{fig:dwkr} does not depend
on values read (other than returning $\perp$ when values differ),  
thus the mapping under the assumption of uniquely written values
remains valid.  Note that an \AReadk\ concurrent with several 
\AWrite\ operations, say $\mathcal{W}^1$, $\mathcal{W}^2$, 
$\mathcal{W}^3$, could read from low-level writes of $\mathcal{W}^1$ 
and low-level writes of $\mathcal{W}^3$, where both 
$\mathcal{W}^1$ and $\mathcal{W}^3$ are effective and write
the same value $v$.  In such a case, it is possible that the 
\AReadk\ returns $v$, picking up some instances of $v$ from 
$\mathcal{W}^1$ and some from $\mathcal{W}^3$.  This does not 
violate quasi-regular behavior, since the value returned would
be due to a concurrent \AWrite.  The planarity argument for 
successive \AReadk\ operations verifies quasi-atomic behavior.     
\end{proof}
\begin{corollary} \label{corollary:dwkr}
In any execution where each \AReadk\ by $p_{i\oplus 1}$ is 
concurrent with at most $(k-1)$ \AWrite\ operations of $p_i$,  
all of $p_{i\oplus 1}$'s {\AReadk}s are quasi-atomic. 
\end{corollary}
\begin{proof}
Any finite execution has some number of $p_i$'s \AWrite\ 
operations, and Lemma \ref{lemma:dwkr}'s graph representation 
of low-level register situations and read operations applies
to this execution.  Each \AReadk\ operation comprises a sequence
of low-level reads, which induces a subgraph for which 
the conditions of Lemma \ref{lemma:dwkr} hold.  Thus, each 
\AReadk\ has quasi-atomic behavior.
\end{proof}

With respect to a single \AWrite, an \ARead(1) 
could return a non-$\perp$ value that
is neither the old (the values of $\textsf{R}_a$ and 
$\textsf{R}_b$ before the \AWrite) nor the new 
value;  instead, the \ARead(1) returns an invalid 
value that we call \emph{contaminated}.  The number of 
contaminated \ARead(1) operations following an \AWrite\ 
is limited, and this fact can be exploited in 
protocols.  The following lemma characterizes 
contamination.    

\begin{lemma} \label{lemma:contam}
In any execution where $p_i$ invokes \AWrite\ at most $m\cdot k$ times,
the number of contaminated \AReadk\ operations is at most $m$.
\end{lemma}
\begin{proof}
Consider the planar graph construction of Lemma \ref{lemma:dwkr} 
representing register situations, applied to the execution from 
the (at most) $m\cdot k$ \AWrite\ operations by $p_i$ and 
some number $t>m$ of \AReadk\ operations invoked by $p_{i\oplus 1}$.   
Looking to find contradiction, suppose $s$ of the \AReadk\ operations
are contaminated, $m<s\leq t$.
For a contaminated \AReadk, in the graph 
all the read operations map to corresponding write operations
representing read concurrent with write, so that these 
read operations return invalid results.  This implies  
that the \AReadk's low-level read
operations map to $k$ distinct vertices, because each
of the $k$ iterations (line 12, Figure \ref{fig:dwkr}) 
scans both $\textsf{R}_a$ and $\textsf{R}_b$, and 
each of these is presumed concurrent with a write to that register. 
The lemma follows because no 
two \AReadk\ invocations have operations mapping to a common 
vertex:  the first \AReadk\ operation is an $R_a$ mapping to
a $w_a$, the last \AReadk\ operation maps to a $w_b$, and 
planarity excludes mapping to common vertices between the first 
and last of these read operations.  From $m\cdot k$   
{\AWrite}s, there are $2m\cdot k$ low-level $w$-vertices, 
and with $s>m$ \AReadk\ operations, there are $s\cdot 2k$ 
$R$-vertices, thus $s>m$ contradicts distinct mapping 
from $R$-vertices to $w$-vertices.
\end{proof}

\section{Two-Register Adaptation of \Dij}
\label{section:tworeg}

The quasi-atomic register protocol of Section \ref{section:qa}
supports transformation of the \Dij\ protocol to the safe register
model.  In the transformed protocol, there are registers
$\textsf{R}_a$ and $\textsf{R}_b$ between each consecutive
pair $p_i$, $p_{i\oplus 1}$, in the ring.  We call the registers
that $p_i$ writes the \emph{output} registers.  
Figure \ref{fig:2rDij} presents the two-register protocol for   
processor $p_i$, $0\leq i<n$.  In this protocol, processors do
not have durable states:  in each cycle of the loop (lines
4-12), processor $p_i$ reads output register $\textsf{R}_a$ into
a local variable (line 4).  The output registers are written by
the \AWrite\ invocation at the end of the cycle (line 12).  
The reading by $p_{i\oplus 1}$ of $p_i$'s output registers occurs
when $p_{i\oplus 1}$ invokes \AReadk\ (line 5).  

\begin{figure}[ht]
\hrule\vspace*{2ex}\par
\begin{tabbing}
XXXXXXXX \= xxx \= xxx \= xxx \= xxx \= xxx \= xxx \= xxx \= \kill
\>\numtiny{1}\> $\Dij_i(\phi,K)$: \\
\>\numtiny{2}\>\> local variables $x$, $y$ \\
\>\numtiny{3}\>\> do \textit{forever} \\
\>\numtiny{4}\>\>\> read from output $\textsf{R}_a$ into $x$ \\
\>\numtiny{5}\>\>\> $y \;\leftarrow\; \ARead(\phi)$ \\
\>\numtiny{6}\>\>\> if $y\neq\perp \;\wedge\; 
		i\neq 0 \;\wedge\; x\neq y$ then \\
\>\numtiny{7}\>\>\>\> $x \;\leftarrow\; y$  \\
\>\numtiny{8}\>\>\>\> \emph{critical section}  \\
\>\numtiny{9}\>\>\> else if $y\neq\perp \;\wedge\; i=0 \;\wedge\; y=x$ then \\
\>\numtiny{10}\>\>\>\> $x \;\leftarrow\; (x+1)\bmod K$ \\
\>\numtiny{11}\>\>\>\> \emph{critical section}  \\
\>\numtiny{12}\>\>\> \AWrite($x$) 
\end{tabbing} 
\hrule
\caption{two register $\Dij(\phi,K)$ protocol for processor $p_i$}
\label{fig:2rDij}
\end{figure}

Two constants need to be set for the protocol, $\phi$ and $K$.  
It is sufficient that $K>2n$, using standard verification 
arguments about \Dij.  Below, we derive a constraint for $\phi$ 
to ensure that \AWrite\ and \ARead($\phi$) invocations 
behave quasi-atomically (a safety property), and later show 
that any sequence of \ARead($\phi$) invocations 
returning $\perp$ is bounded (a progress property).  For the
following lemma, an $\AWrite(x)$ invocation is 
called \emph{effective} in case $x$ differs from the value of 
the output registers;  an ineffective write merely reads the
output registers, finding they already contain $x$, and returns.   

\begin{lemma} \label{lemma:nobogus}
In any execution of the protocol of Figure \ref{fig:2rDij} with
$\phi>2n$, no invocation of \ARead($\phi$) has a contaminated response.
\end{lemma}
\begin{proof}
The lemma is shown by contradiction, assuming that in some 
execution $E$ there is an \ARead($\phi$) by processor $p_i$ 
returning a contaminated value.  The contradiction is demonstrated
by deducing that the contaminated \ARead($\phi$) at $p_i$ is concurrent
with an \AWrite\ invocation, also at $p_i$ (which is impossible
because no processor concurrently invokes both \ARead\ and
\AWrite). 

To set up the contradiction, we consider the first contaminated 
\ARead($\phi$) in $E$, occurring at processor $p_i$, and apply 
Lemma \ref{lemma:dwkr} to infer that processor $p_{i\ominus 1}$ 
invoked at least $\phi$ effective {\AWrite}s, so that each of the 
\ARead($\phi$)'s register operations was concurrent with a 
corresponding write by $p_{i\ominus 1}$.  Figure \ref{fig:situation0}  
depicts the situation, where the vertical dotted lines indicate concurrent
read and write operations; for instance, $r^1_a$ and $w^1_a$ are concurrent.

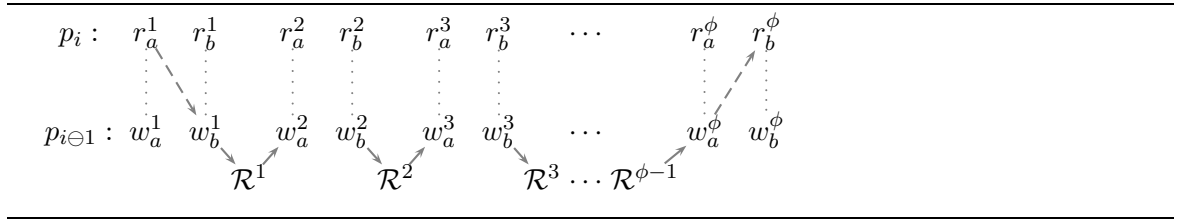
\begin{figure}[ht]
\hrule\vspace*{2ex}\par
\begin{pspicture}(-0.5,-0.5)(8,2)
\begin{psmatrix}[rowsep=5pt,colsep=3pt]
$p_i:$ & \quad & $r^1_a$ & ~ & $r^1_b$ 
		  & ~ & $r^2_a$ & ~ & $r^2_b$ 
		  & ~ & $r^3_a$ & ~ & $r^3_b$ 
                  & ~ & $\cdots$     
		  & ~ & $r^{\phi}_a$ & ~ & $r^{\phi}_b$ \\ \\ 
$p_{i\ominus 1}:$ & & $w^1_a$ & ~ & $w^1_b$ 
		  & ~ & $w^2_a$ & ~ & $w^2_b$ 
		  & ~ & $w^3_a$ & ~ & $w^3_b$ 
                  & ~ & $\cdots$     
		  & ~ & $w^{\phi}_a$ & ~ & $w^{\phi}_b$ \\ 
     & & & & & $\mathcal{R}^1$ 
		& & & & $\mathcal{R}^2$ 
                & & & & $\mathcal{R}^3$ 
                & $\cdots$ & $\mathcal{R}^{\phi-1}$ 
\psset{linestyle=dotted,linecolor=lightgray}
\ncline{3,3}{1,3}\ncline{3,5}{1,5}\ncline{3,7}{1,7}\ncline{3,9}{1,9} 
\ncline{3,11}{1,11}\ncline{3,13}{1,13}\ncline{3,17}{1,17}\ncline{3,19}{1,19}
\psset{linestyle=solid,arrows=->}
\ncline{3,5}{4,6}\ncline{4,6}{3,7}\ncline{3,9}{4,10}\ncline{4,10}{3,11}
\ncline{3,13}{4,14}\ncline{4,16}{3,17}
\psset{linestyle=dashed,arrows=->}
\ncline{1,3}{3,5}\ncline{3,17}{1,19}
\end{psmatrix}
\end{pspicture}
\hrule
\caption{situation for $p_i$ and $p_{i\ominus 1}$}
\label{fig:situation0}
\end{figure}

The figure labels $p_{i\ominus 1}$'s write operations $w^1_a$, $w^1_b$, and 
so on, however it may be that $w^i_a$ and $w^i_b$ do not belong to the same
\AWrite.  The figure is thus unlike the labeling of (\ref{eqn:awrite1}),  
because the labeling $w^1_a$, $w^1_b$, $w^2_a$, \ldots, $w^{\phi}_b$ 
comprise a subsequence of low-level register writes selected for the 
counterexample, to be concurrent with read operations.   There could, 
in fact, be numerous effective \AWrite\ operations 
between $w^i_b$ and $w^{i+1}_a$.  The figure also shows some 
\ARead\ invocations by $p_{i\ominus 1}$, labeled as 
$\mathcal{R}^1$, \ldots, $\mathcal{R}^{\phi-1}$.  This follows from 
the logic of the protocol in Figure \ref{fig:2rDij}, in which 
any \AWrite\ at line 12 is followed by an \ARead\ on line 5.
The arrows between $w$ and $\mathcal{R}$ items in
the figure signify precedence:  
$w^1_b\prec \mathcal{R}^1$, for example.  The 
dashed arrow from $r^1_a$ to $w^1_b$ represents $r^1_a\preceq w^1_b$, 
which holds because $r^1_a$ must end before $w^1_b$ ends so that 
$r^1_b$ can be concurrent with $w^1_b$.  Just as there could be 
numerous {\AWrite}s between successive $w$-vertices in the figure,    
there could be other \ARead\ invocations by $p_{i\ominus 1}$ 
not shown in the figure:  there could be invocations 
that do not return values which would result in 
effective \AWrite\ invocations by $p_{i\ominus 1}$.  
One more observation about the situation of Figure 
\ref{fig:situation0} concerns planarity:  though the low-level 
$w$ and $r$ instances shown may be selected subsequences induced by 
\ARead\ and \AWrite\ operations, the graph of the figure is planar, 
by arguments similar to those given in the proof of Lemma \ref{lemma:dwkr}.
Below, this planarity is implicitly used in arguments about the 
transitivity of precedence. 

A next step in the proof is a deduction about \AWrite\ 
and \ARead\ invocations at $p_{i\ominus 2}$, many of which
are concurrent with the scenario of Figure \ref{fig:situation0};  
a similar deduction can establish concurrency with \AWrite\ and
\ARead\ invocations at $p_{i\ominus 3}$;  more generally, there
is a chain of deductions about concurrency of operations. 
To construct this chain of deductions, 
we depict the scenario
between $p_{i\ominus t}$ and $p_{i\ominus (t+1)}$ in 
Figure \ref{fig:situation1}.

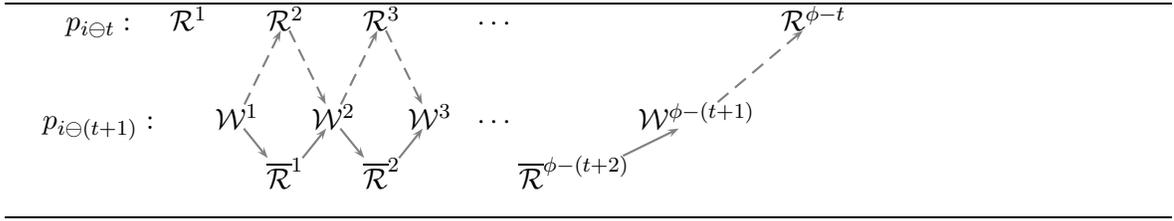
\begin{figure}[ht]
\hrule\vspace*{2ex}\par
\begin{pspicture}(-0.5,-0.5)(8,2)
\begin{psmatrix}[rowsep=5pt,colsep=3pt]
$p_{i\ominus t}:$ & \quad & $\mathcal{R}^1$ & ~ 
                & $\mathcal{R}^2$ & ~ & $\mathcal{R}^3$ 
		& ~ & ~ & $\cdots$ & ~ & ~ & ~ 
		& $\mathcal{R}^{\phi-t}$ \\ \\  
$p_{i\ominus(t+1)}:$ & & & 
        $\mathcal{W}^1$ & & $\mathcal{W}^2$ & & $\mathcal{W}^3$ 
	& & $\cdots$ & & $\mathcal{W}^{\phi-(t+1)}$ \\
	& & & & $\overline{\mathcal{R}}^1$ & & 
        $\overline{\mathcal{R}}^2$ & & & & 
        $\overline{\mathcal{R}}^{\phi-(t+2)}$ 
\psset{linestyle=solid,linecolor=lightgray,arrows=->}
\ncline{3,4}{4,5}\ncline{4,5}{3,6}
\ncline{3,6}{4,7}\ncline{4,7}{3,8}
\ncline{4,11}{3,12}
\psset{linestyle=dashed,arrows=->}
\ncline{3,4}{1,5}
\ncline{3,6}{1,7}
\ncline{3,12}{1,14}
\ncline{1,5}{3,6}
\ncline{1,7}{3,8}
\end{psmatrix}
\end{pspicture}
\hrule
\caption{situation for $p_{i\ominus t}$ and $p_{i\ominus (t+1)}$}
\label{fig:situation1}
\end{figure}

In Figure \ref{fig:situation1}, processor $p_{i\ominus t}$'s first 
\ARead, labeled $\mathcal{R}^1$, is presumed to be a reading of the 
initial registers before $p_{i\ominus(t+1)}$ has written them: 
we suppose this to obtain
the worst case (fewest number of effective {\AWrite}s) for 
$p_{i\ominus(t+1)}$'s behavior.  Thus the first \ARead\ at 
$p_{i\ominus t}$ influenced by $p_{i\ominus(t+1)}$ is $\mathcal{R}_2$, and the
dashed arrow from $\mathcal{W}^1$ to $\mathcal{R}^2$ 
indicates that $\mathcal{W}^1\preceq \mathcal{R}^2$;  
also $\mathcal{R}^2\preceq \mathcal{W}^2$ is 
represented by a dashed arrow, since $\mathcal{R}^2$ 
gets the value written by 
$\mathcal{W}^1$ (and not by $\mathcal{W}^2$, 
because it cannot be that
$\mathcal{W}^2\prec \mathcal{R}^2$).  
The first \AWrite\ $\mathcal{W}^1$ need 
not be preceded by an \ARead\ at $p_{i\ominus(t+1)}$, 
because the initial state of $E$ is arbitrary.  
The {\ARead}s of processor $p_{i\ominus(t+1)}$ are  
denoted as $\overline{\mathcal{R}}$-vertices.

\begin{observation} \label{obs:contain}
Containment properties
accompanying the definitions of $\prec$ and $\preceq$ relations 
enable the following assertion:  an interval 
from some point in $\mathcal{R}^2$ through
some point in $\mathcal{R}^{\phi-t}$ 
contains the interval beginning from the end of $\mathcal{W}^2$ 
through the start of $\mathcal{W}^{\phi-(t+1)}$,  
which contains the interval of $p_{i-(t+1)}$ 
from $\overline{\mathcal{R}}^2$ through
$\overline{\mathcal{R}}^{\phi-(t+2)}$.
\end{observation}

Let $I_t$ denote the interval from $\mathcal{R}^2$ through 
$\mathcal{R}^{\phi-t}$.  Interval $I_{t+1}$ thus goes 
Observation \ref{obs:contain} can
be restated as:  interval $I_t$ contains $I_{t+1}$.  By 
transitivity and a simple induction, interval $I_1$ contains
$I_t$ for $2\leq t<\phi/2$ (each step of the induction decreases
the number of terms by 2).
Therefore, if $\phi\geq 2n$, we deduce that $I_1$ contains $I_n$, which
is an interval of $p_{i\ominus n}=p_i$.  That is the linchpin of 
the proof's argument:  the contradicting scenario implies that $p_i$'s 
reading of a contaminated variable depends on $p_i$ injecting the 
contamination, which would have to continue around the ring.    
In particular, for line 5's \ARead\ to return a contaminated 
value at $p_i$, at least one register read by $p_i$ would have to
be concurrent with a register write by $p_i$ due to the 
\AWrite\ of statement 12, which is not possible.  The assumption 
of a contaminated result at line 5 is thereby contradicted, provided
$\phi\geq 2n$.  
\end{proof}

\begin{corollary}
In any execution of the protocol of Figure \ref{fig:2rDij} with
$\phi>2n$, every invocation of \ARead($\phi$) has quasi-atomic behavior. 
\end{corollary}
\begin{proof}
Corollary \ref{corollary:dwkr} establishes the conditions for quasi-atomic
behavior:  if $p_{i\ominus 1}$ invokes \AWrite\ at most $(\phi-1)$ times 
between each of $p_i$'s \ARead($\phi$) operations, then $p_i$'s 
{\ARead}s are quasi-atomic.  Arguments given in Lemma \ref{lemma:nobogus}'s
proof show, by contradiction, that $p_{i\ominus 1}$ cannot have 
$\phi>2n$ effective \AWrite\ operations concurrent with an \ARead($\phi$) 
by $p_i$.  Any \AWrite\ operations not
concurrent with $p_i$'s \ARead($\phi$) have no effect on quasi-atomicity,
as was explained in the proof of Lemma \ref{lemma:dwkr}.     
\end{proof}

\begin{lemma} 
In any execution of the protocol of Figure \ref{fig:2rDij} with
$\phi>2n$, the number of consecutive $\perp$ responses for any 
$p_i$ at line 5 is bounded.
\end{lemma}
\begin{proof}
We first show, by contradiction, that no execution can have 
\emph{all} \ARead\ operations return $\perp$:
if all \ARead($\phi$) operations return
$\perp$, then eventually the value of $x$ in Figure \ref{fig:2rDij} 
remains constant for each $p_i$, throughout the execution.  Thus
no \AWrite\ operation is effective, and    
no registers are written throughout the 
execution.  Thereafter, every \ARead($\phi$) encounters no 
concurrent \AWrite;  but this implies all low-level reads by 
any $p_i$ obtain the same value, which contradicts the assumed 
return of $\perp$ shown in Figure \ref{fig:dwkr}.

Now, again by contradiction, we show that no particular $p_i$'s
\ARead\ operations continually return $\perp$.  
If $p_i$ forever returns $\perp$, then 
eventually $p_{i\oplus 1}$ has no effective \AWrite\ operations; 
by induction going around the ring, it follows that 
$p_{i\ominus 1}$ eventually has no effective
\AWrite\ operations.  This contradicts conditions of returning 
$\perp$ in Figure \ref{fig:dwkr}.
\end{proof}

\begin{theorem}
If $\phi>2n$ and $K>2n$, then 
the two-register adaptation of 
$\Dij(\phi,K)$ given in Figure \ref{fig:2rDij} 
is self-stabilizing to mutual exclusion. 
\end{theorem}
\begin{proof}
Having shown that \ARead($\phi$) has quasi-atomic behavior and 
the absence of deadlock (e.g., no $p_i$ continually encounters 
$\perp$ values for \ARead\ operations), 
the standard convergence arguments for 
\Dij\ apply:  $K>2n$ implies that eventually $p_0$ 
obtains a value $x$ that exists
nowhere else in the ring, and this is enough to enforce 
convergence to mutual exclusion.
\end{proof}

\section{$O(\lg n)$-Register Adaptation of \Dij}
\label{section:logreg}

When processor communication using registers and execution is 
asynchronous, the number of reads by $p_i$ from $p_{i\ominus 1}$'s 
output registers per effective write is unbounded:  $p_i$ could
be unboundedly faster than $p_{i\ominus 1}$, hence many reads 
get no new information.  Such scenarios are unavoidable, however
the \Dij\ protocol of Section \ref{section:tworeg} uses many 
reads per effective write even in the best case, because 
$\ARead(\phi)$ scans input registers at least $2n$ times.  The
point of this section is to introduce another \Dij\ adaptation
scans input registers $O(\lg n)$ times in the best case.  This 
can be achieved using $\ARead(2)$ and $O(\lg n)$ registers between
each pair $(p_i,p_{i\oplus 1})$ of processors.  The basis of
the construction is an idea introduced in \cite{DH01}, which uses
a gray code \cite{Gar86} representation of the token.  
The improvement here is a protocol that is simpler to reason about 
than the algorithm of \cite{DH01}, which instead introduces a parity bit 
manipulated in each write operation, and lacks the formal 
structure that Lemma \ref{lemma:contam} provides.

Figure \ref{fig:nrDij} presents the protocol.  Each processor 
$p_i$ writes to an array of registers, managed by the \AWrite/\ARead\ 
construction of Section \ref{section:qa}.  The constant $k$ 
specifies the number of register pairs $(\textsf{R}_a[i],\textsf{R}_b[i])$,
for $0\leq i<k$.  The register pair for index $i$ corresponds to the 
$i^\textrm{th}$ bit in the gray code representation of a token value.    
For arguments about the protocol, let 
$\textsf{R}_{a/b}[i]$ denote the register pair for bit $i$.

The invocation $\ARead_i(2)$ specifies an \ARead(2) invocation
on input pair of registers for bit $i$;  
$\AWrite_i(val)$ similarly specifies the 
output register pair to use for writing.  Function 
$\textsf{gray}_k^{-1}$ used on lines 7 and 10 decodes the $k$-bit 
gray code representation of a non-negative integer; for line 10,     
$\textsf{gray}_k^{-1}$ may encounter a $\perp$ value for one or more
bits.  The convention for such cases is that $\textsf{gray}_k^{-1}$
maps to $\perp$ if $\ARead_i(2)$ returns $\perp$ for any $i$.  

Three iterations process registers, seen on lines 5, 8, and 18.  
Whereas the iterations of lines 5 and 8 go from 0 to $k-1$,  
the iteration of line 18 goes in the reverse order:  this is 
intentional, and simplifies reasoning about the atomicity of 
token transfer in a proof.  

\begin{figure}[ht]
\hrule\vspace*{2ex}\par
\begin{tabbing}
XXXXXXXX \= xxx \= xxx \= xxx \= xxx \= xxx \= xxx \= xxx \= \kill
\>\numtiny{1}\> $\Dij_i(K)$: \\
\>\numtiny{2}\>\> constant $k = \lceil \lg K\rceil$ \\
\>\numtiny{3}\>\> local variables $X[k]$, $Y[k]$, $x$, $y$ \\
\>\numtiny{4}\>\> do \textit{forever} \\
\>\numtiny{5}\>\>\> for $i\in 0..(k-1)$ \\
\>\numtiny{6}\>\>\>\> read from output $\textsf{R}_a[i]$ into $X[i]$ \\
\>\numtiny{7}\>\>\> $x \leftarrow \textsf{gray}^{-1}_k(X)\bmod K$ \\
\>\numtiny{8}\>\>\> for $i\in 0..(k-1)$ \\
\>\numtiny{9}\>\>\>\> $Y[i] \leftarrow \ARead_i(2)$ \\
\>\numtiny{10}\>\>\> $y \;\leftarrow\; \textsf{gray}^{-1}_k(Y)\bmod K$ \\
\>\numtiny{11}\>\>\> if $y\neq\perp \;\wedge\; 
		i\neq 0 \;\wedge\; x\neq y$ then \\
\>\numtiny{12}\>\>\>\> $x \;\leftarrow\; y$  \\
\>\numtiny{13}\>\>\>\> \emph{critical section}  \\
\>\numtiny{14}\>\>\> else if $y\neq\perp \;\wedge\; i=0 \;\wedge\; y=x$ then \\
\>\numtiny{15}\>\>\>\> $x \;\leftarrow\; (x+1)\bmod K$ \\
\>\numtiny{16}\>\>\>\> \emph{critical section}  \\
\>\numtiny{17}\>\>\> $X \leftarrow \textsf{gray}_k(x)$ \\
\>\numtiny{18}\>\>\> for $i\in (k-1)..0$ \\
\>\numtiny{19}\>\>\>\> $\AWrite_i(X[i])$ 
\end{tabbing} 
\hrule
\caption{two register $\Dij(\phi,K)$ protocol for processor $p_i$}
\label{fig:nrDij}
\end{figure}

\begin{wrapfigure}{r}{0.3\columnwidth}
\begin{center}\framebox{
\begin{tabular}{cl}
\textsf{value} & \textsf{bits} \\ \hline
0 & 000 \\
1 & 001 \\
2 & 011 \\
3 & 010 \\
4 & 110 \\
5 & 111 \\
6 & 101 \\
7 & 100 
\end{tabular}}\end{center}
\caption{3-bit gray code}
\label{fig:graycode}
\end{wrapfigure}

The validation of the protocol builds on some simple properties and 
on the definition of a certain type of state in an execution.  Recall
that gray code, like binary arithmetic, orders the bits of its 
representation in order from most significant to least significant. 
Figure \ref{fig:graycode} shows a 3-bit reflected gray code, for
example.  

For the local variables defined on line 3 of Figure \ref{fig:nrDij}, 
and for the register pair $\textsf{R}_{a/b}[i]$, 
the most significant bit (MSB) has the least index.  
Thus $\textsf{R}[k-1]$ represents the least 
significant bit (LSB).  Like standard binary encoding, in a sequence of 
increments of a gray code value, the LSB alternates more frequently
than does the MSB:  $2^k-1$ increments to a $k$-bit 
gray code changes the LSB $2^{k-2}$ times (repeating the sequence of
two 0's, followed by two 1's), whereas the MSB changes only twice.
A useful property of the gray code is that each increment changes 
exactly one bit in the encoding (including rollover from the 
largest representable integer).   

We define a \emph{flash state} to be one where all values for the 
MSB, in any register or any internal variable of 
any processor, are zero.  A \emph{flash event} is the transition from
a flash state to a non-flash state.  A flash event only occurs by 
the step $x \;\leftarrow\; (x+1)\bmod K$ in line 15 of the protocol. 
After a flash event, $p_0$ writes the unique one-valued MSB in the ring.  
A \emph{home state} is one where all values for 
all bits and corresponding internal variables are equal in corresponding
bit positions (different bits may have different values, however a bit
at any position has the same value everywhere).  A legitimate state 
for the protocol is either a home state or a successor of a legitimate
state.  

Some elementary properties of executions originating from a
home state are \itp{i} a home state is reached infinitely often, 
and \itp{ii} all effective \AWrite\ operations are atomic.  Properties
\itp{i}--\itp{ii} can be shown by induction, paralleling standard 
arguments for the \Dij\ protocol.  Thanks to property \itp{i} and 
the definition of a legitimate state, validation of the protocol in 
Figure \ref{fig:nrDij} consists of showing that any execution eventually
reaches a home state.  Property \itp{ii} is technical statement about
the conditions of write and effective \AWrite\ operations:  at most
one processor can be engaged in an effective \AWrite\ at any time 
in an execution of legitimate states, and following the completion of an   
\AWrite\ by $p_i$, processor $p_{i\oplus 1}$ correctly reads the value
before the next effective \AWrite.  The gray coding ensures that only 
one \AWrite\ can be effective in the iteration of lines 18-19 of the
protocol.  

In a legitimate state, a register pair 
$\textsf{R}_{a/b}[i]$ 
are equal except during an \AWrite\ operation, 
which may have written $\textsf{R}_a$ but not yet $\textsf{R}_b$.  
With respect to any state in an execution, 
a register pair is said to be \emph{coherent} if both registers
have the same value or an \AWrite\ operation is underway.
Observe that the procedure defining \AWrite\ in Figure \ref{fig:dwkr} 
ensures that both registers are equal upon completion, whether or not the 
\AWrite\ is effective.  Thus, in any execution, 
after each processor has performed all the 
steps in lines 18-19 of the protocol, it follows that all 
register pairs are coherent for all subsequent states.  

\begin{lemma} 
\label{lemma:gohome}
Any execution starting from a flash state contains a home state.
\end{lemma}
\begin{proof}
We focus on $p_0$'s behavior for the proof.  Only $p_0$ is capable
of changing its most significant bit from zero to one, by the  
assignment of line 15.  All other processors copy input register 
values to output register values.  The proof of the lemma is in 
two parts:  first, we show that $p_0$ eventually does change the 
MSB, that is, that a flash event occurs; 
the second part is to show that a 
home state is reached sometime after the flash event. 
\par
The inevitability of a flash event is shown by contradiction.  
Suppose $p_0$ never changes its most significant bit.  After some
writes of other bits, $p_0$ has no effective writes throughout
some suffix of the execution, because line 15 does not execute 
infinitely often by assumption.  It follows that eventually there
is a suffix where $p_1$'s output registers have the same values
as $p_0$'s output registers, as $p_1$ will copy these values in 
some cycle of the protocol (line 12) --- there cannot be a 
$\perp$-value read when there is no concurrent write by $p_0$.  
By induction, $p_i$ for $0<i<n$ eventually also has the same 
output registers as $p_0$, and no processor will have any 
effective write for the remainder of the execution.  
However, such a condition contradicts   
the condition of line 14 for processor $p_0$, implying that 
a flash event must occur.
\par
A flash event has $p_0$ assigning one to the MSB, thus writing 
$\textsf{R}_a[0]\leftarrow 1$ and $\textsf{R}_b[0]\leftarrow 1$.    
After the \AWrite\ operation at $p_0$ associated with this 
flash event, the MSB of $p_0$ is the only MSB with 1.  In fact,
$p_0$ will not again perform an effective write until this 1
value propagates through the ring (for instance, $p_{n-1}$ 
has 0 for the MSB, and does not engage in an effective write
until it copies 1 from $p_{n-2}$).   Consider the event of 
$p_1$ reading the 1 MSB from $p_0$ by an \ARead(2) operation. 
This \ARead\ has quasi-atomic behavior because the two low-level
writes to $\textsf{R}_a$ and $\textsf{R}_b$ of a single 
\AWrite\ by $p_0$ cannot be concurrent with all four low-level
reads of the \ARead\ operation.  Furthermore, all the \AWrite\ 
operations to less significant bit positions occur \emph{before}
the \AWrite\ of the MSB, which implies that after $p_1$ reads 
1 for the MSB, all the other bits that $p_1$ reads are atomic
and have the values written by $p_0$.  Inductively, this 
argument holds for the transfer of values from $p_i$ to 
$p_{i\oplus 1}$, up to $p_{n-1}$.  Finally, after $p_{n-1}$ 
writes 1 for its MSB, we infer that all values at all positions
are the same throughout the ring, which is a home state. 
\end{proof}
\begin{lemma} 
\label{lemma:goflash}
Any execution contains a flash state.
\end{lemma}
\begin{proof}
Using arguments (based on contradiction) similar to those 
in the proof of the previous lemma, $p_0$ executes line 15 
infinitely often in any execution, so the MSB of $p_0$ changes
throughout the execution.  To show that a flash state occurs, 
we consider $p_0$ invoking an effective \AWrite(0) and deduce
that $p_1$ copies its MSB from $p_0$, then $p_2$ copies its
MSB from $p_1$, and generally $p_{i\oplus 1}$ copies from 
$p_i$, all before $p_0$ invokes \AWrite(1);  this shows that
a flash state is reached, provided the copying of MSBs occurs
in sequence, so that all are zero valued.
\par
After each token increment (line 15), $p_0$ writes the token 
value to output registers and waits until the same value is 
read from $p_{n-1}$.  A property of the gray code is that  
the LSB changes in half of the token increments.  Since $2^k\geq K$, 
the LSB changes at least $K/2>n$ times between consecutive 
\AWrite(0) and \AWrite(1) operations of the MSB.  Put another way,
$p_0$ expects to observe at least $n$ changes of the LSB in 
this period.  The question is, which of these changes are due to 
contaminated reads (\emph{e.g.,} an \ARead\ at $p_0$ concurrent with multiple
\AWrite\ operations by $p_{n-1})$, which are due to LSB values 
initially present in processors other than $p_0$, and which are 
values propagated around the ring, from $p_0$ back to $p_0$.       
By counting these types of changes, we shall bound the number 
of values not propagated around the ring, showing them to be 
at most $n$ in total.
\par
Suppose $p_0$ does not write any registers after the 
\AWrite(0) of the MSB completes;  we count the number of 
LSB changes that $p_0$ \emph{could} observe during the subsequent
execution.  The count is derived inductively, starting with
the number of LSB values observed by $p_1$.  The case for $p_1$ 
is simple because we suppose $p_0$ writes once.  Processor 
$p_1$ may observe an initial value, and then another value that
$p_0$ writes.  We ignore the case of reading $\perp$, because
the protocol of Figure \ref{fig:nrDij}.  Since the LSB is written
at most once by $p_0$, each \ARead(2) by $p_1$ is atomic, so
no contaminated reading occurs.  The conclusion is that $p_1$ 
observes at most two values for the LSB.  Each such observed value 
at $p_1$ could result in an effective \AWrite\ of its LSB. 
\par
Counting the observable values for $p_2$ introduces contaminated
values:  because $p_1$ may write the LSB register pair twice, 
$p_2$ could read a contaminated value, however, Lemma \ref{lemma:contam}
limits to one the number of contaminated reads.  If $p_2$ does
read a contaminated value, it follows that the correct value would
be observed by another \ARead(2).  Another scenario for $p_2$ is 
the absence of contaminated values, in which case $p_2$ may observe
both values written by $p_1$.  The total number of observable values
is three in either scenario:  one for the initial value, followed
by two more observed values due to $p_1$'s writes.   
\par
The induction hypothesis is that $p_i$ may observe at most $i+1$ 
values in the execution where $p_0$ does not write any registers.
Assume that $p_{i-1}$ observes and writes at most $i$ values for
the LSB.  As $p_i$ reads the values it is possible that some (or 
all) of the writes are concurrent with $p_i$'s \ARead(2) operations,
resulting in contaminated reads.  Again, Lemma \ref{lemma:contam} 
limits the number of contaminated values to be at most half the
number of \AWrite\ operators by $p_{i-1}$.  It follows that in 
any scenario, $p_i$ observes at most $i$ values due to $p_{i-1}$'s
writes.  The total number is $i+1$ because $p_i$ can also observe
the initial value of the MSB.    
\par
The conclusion from the induction is that $p_0$ ``observes'' at most 
$n$ changes to the LSB read from $p_{n-1}$ (these would not be actually 
observed because we suppose $p_0$ does not write any registers).   
Note that if all $n$ changes due to initial values and operations 
by $p_1$--$p_{n-1}$ without influence of $p_0$ are observed first at 
$p_0$, before any influence of values written by $p_0$ circulate the
ring, then the MSB at $p_0$ retains the value 0, because more than 
$n$ changes of the LSB are needed to enable \AWrite(1) of the MSB.   
It remains to consider more rapid influence of values written by 
$p_0$ affecting what other processors write.  Any values copied 
directly or indirectly from $p_i$ to $p_{i+1}$ do so only for 
non $\perp$-ARead(2) operations;  and since $p_0$ writes the MSB 
once in the execution under examination, it follows that any 
such copying obtains the value 0 for the MSB.  Therefore, after
$n$ changes to the LSB by $p_{n-1}$, the next change of the LSB
is due to a value circulating the ring, from $p_0$ to $p_{n-1}$.  
Each \ARead(2) operation influenced by $p_0$ values includes 
an atomic reading of the MSB copied from $p_0$, hence the 
$(n+1)^{\textrm{\scriptsize th}}$ change to the LSB is accompanied, if not 
preceded, by $p_{n-1}$ writing 0 to its MSB output pair.  This
establishes a flash state. 
\end{proof}

\begin{theorem}
If $K>2n$, then 
the $O(\lg K)$-register adaptation of 
$\Dij(K)$ given in Figure \ref{fig:nrDij} 
is self-stabilizing to mutual exclusion. 
\end{theorem}
\begin{proof}
Every execution of the protocol has a suffix in which all
states have coherent registers.  Within such a suffix, 
Lemma \ref{lemma:goflash} is applicable, guaranteeing that 
a flash state eventually occurs.  Subsequently, Lemma 
\ref{lemma:gohome} asserts that a home state will be reached,
whereafter registers behave atomically, because at each 
state the choice of what register pair will next be effectively
written is deterministic, and once the {\AWrite}s of lines 18--19
complete, the result will be atomically read before the next
effective write is enabled.  Thus, standard arguments for \Dij\ 
apply to show safety.   
\end{proof}

\section{Discussion}
\label{section:discuss}

The protocols of Section \ref{section:qa} use well known techniques
for register constructions:  duplicating written values and multiple
scans by the reader are standard fare in the literature.  The 
adaptation in Section \ref{section:tworeg} takes advantage of inherent
limitations on register writing, even for an illegitimate state, of
the \Dij\ protocol.  Section \ref{section:logreg} exploits two more
standard techniques from the literature of register constructions, 
representing a value with bit registers (where safe and regular 
properties coincide) and the idea of ordering writes and reading
scans in opposite directions \cite{Lam86b,Vid88}. 

Ideas for limiting concurrency, particularly in common shared memory
models, include counting or balancing networks and filters in 
mutual exclusion algorithms.  However the technique use here is different, 
being geared to the \Dij\ protocol.  One might therefore consider the 
protocols of this paper to be of very limited use in other contexts.  
However, the history of self-stabilization literature should be 
consulted before such a judgment.  Generalizations of the token ring
lead to wave protocols (propagation of information with feedback), 
and other synchronization or control algorithms.  Several of the 
crucial properties of \Dij\ are enjoyed by other self-stabilizing 
(and non-stabilizing) protocols, including implicit restrictions on 
concurrency.  For instance, for many protocols, quiescence of selected
processes results in deadlock, so there is hope 
that counter-flushing \cite{Var00}
or similar techniques could simplify the adaptation to safe-register 
communication.

There have been relatively few investigations of wait-free self-stabilization
or stabilization in the common shared memory model:  papers appear 
sporadically over the years since \Dij\ first appeared \cite{Lam86d,HPT95,AADDGT10}.
This intersection of topics appears to contain many unresolved questions.


\begin{thebibliography}{10}

\bibitem{Dij73}
EW Dijkstra,
EWD391 Self-stabilization in spite of distributed control.
In \emph{Selected Writings}, pages 41--46, 
Springer-Verlag, 1982 (original date is 1973; 
printed in 1982).

\bibitem{Dij74} EW Dijkstra,
Self stabilizing systems in spite of distributed control.
\emph{Communications of the ACM}, 17:643--644, 1974.

%\bibitem{La77} L Lamport,
%Concurrent reading and writing,
%\emph{Communications of the ACM}, 20(11):806--811, 1977.

\bibitem{Hoa78} CAR Hoare, 
Communicating sequential processes.
\emph{Communications of the ACM} 21(8):666-677, 1978.

\bibitem{Gar86} M Gardner, 
\emph{Knotted Doughnuts}, chapter 2: the
binary gray code. Pages 11--27, W H Freeman
and Company, 1986.

\bibitem{Lam86a} L Lamport,
On interprocess communication, part I: basic formalism.
\emph{Distributed Computing} 1(1):77--85, 1986.

\bibitem{Lam86b} L Lamport,
On interprocess communication, part II: algorithms.
\emph{Distributed Computing} 1(1):86--101, 1986.

\bibitem{Lam86c} L Lamport, 
The mutual exclusion problem: part I---a theory of 
interprocess communication.
\emph{Journal of the ACM} 33(2):313-326, 1986.

\bibitem{Lam86d} L Lamport, 
The mutual exclusion problem: part II---statement
and solutions. 
\emph{Journal of the ACM}, 33(2):327--348, 1986.

\bibitem{Mis86} J Misra,
Axioms for memory access in asynchronous hardware systems. 
\emph{ACM Transactions on Programming Languages and Systems}, 8(1):142-153,
1986.

\bibitem{Vid88} K Vidyasankar,
Converting Lamport's regular register to atomic register.
\emph{Information Processing Letters} 28:287-290, 1988.

\bibitem{HPT95} JH Hoepman, M Papatriantafilou, P Tsigas,
Self-stabilization of wait-free shared memory objects.
In \emph{Proceedings of the 9th International Workshop
on Distributed Algorithms} (WDAG95), Springer LNCS 972, pages 273-287,
1995.

\bibitem{Var00} G Varghese,
Self-stabilization by counter flushing.
\emph{SIAM Journal on Computing} 30(2):486--510, 2000. 

\bibitem{Dol00} S Dolev, \emph{Self-stabilization}, MIT Press, 2000.

\bibitem{DH01} S Dolev, T Herman,  
Dijkstra's self-stabilizing algorithm in unsupportive environments.
In \emph{Proceedings of the Fifth International 
Workshop on Self-Stabilizing Systems} (WSS2001), Springer LNCS 2194, 
pages 67-81, 2001.

\bibitem{AADDGT10} N Alon, H Attiya, S Dolev, S Dubois, M Gradinariu, S Tixeuil,
Brief announcement: sharing memory in a self-stabilizing manner.  
In \emph{Proceedings of the 24th International Symposium on Distributed Computing}
(DISC10), Springer LNCS 6343, pages 525-527, 2010. 

\end{thebibliography}
\end{document}